\documentclass[11pt,a4paper]{article}

\usepackage{amsmath,amssymb,amsthm,hyperref,color}
\usepackage{filecontents}
\usepackage{enumerate}
\hypersetup{colorlinks=true,linkcolor=blue, citecolor=blue}
\usepackage{enumerate}

\usepackage{a4wide}

\newtheorem{theorem}{Theorem}
\newtheorem{lemma}[theorem]{Lemma}

\newtheorem{definition}[theorem]{Definition}

\title {The complexity of counting locally maximal 
    satisfying assignments of Boolean CSPs\thanks{  The research leading to these results has received funding from
    the European Research Council under the European Union's Seventh
    Framework Programme (FP7/2007--2013) ERC grant agreement
    no.\ 334828. The paper reflects only the authors' views and not
    the views of the ERC or the European Commission. The European
    Union is not liable for any use that may be made of the
    information contained therein.}}

\author{Leslie Ann Goldberg\thanks{
Department of Computer Science, University of Oxford,   UK.  
} \and Mark Jerrum\thanks{
School of Mathematical Sciences, Queen Mary, University of
London,   UK. } 
 }
  
\newcommand\FP{\ensuremath{\mathsf{FP}}}

\newcommand{\numP}{\ensuremath{\mathsf{\#P}}}

\newcommand{\numCSP}[1]{\#\mathsf{CSP}(#1)}
\newcommand{\CSP}[1]{\mathsf{CSP}(#1)}
\newcommand{\numMaxCSP}[1]{\#\mathsf{LocalMaxCSP}(#1)}
\def\SAT{\#\mathsf{SAT}}
\def\BIS{\#\mathsf{BIS}}

\newcommand{\IMtwo} {\mathsf{IM}_2}
\newcommand{\APeq}{\equiv_{\mathsf{AP}}}
\newcommand{\APred}{\leq_{\mathsf{AP}}}

\newcommand{\tospin}[3]{{#1}_{[{#2} \rightarrow {#3}]}}
\newcommand{\zeroes}[1]{\mathcal{Z}(#1)}
\newcommand{\nonzeroes}[1]{\mathcal{N}(#1)}
\newcommand{\maxzero}[2]{\mathcal{M}_0(#1,#2)}
\newcommand{\maxone}[2]{\mathcal{M}_1(#1,#2)}
\newcommand{\bad}[2]{\mathcal{B}(#1,#2)}
\newcommand{\OR}{\mathsf{OR}}
\newcommand{\NAND}{\mathsf{NAND}}    
  
\newcommand{\Implies}{\mathsf{Implies}}
\newcommand{\Uzero}{\mathsf{U}_0}
\newcommand{\Uone}{\mathsf{U}_1}
\newcommand{\alphabet}{\Sigma}
\newcommand{\poly}{\mathop{\mathrm{poly}}}
\newcommand{\circum}{\#\mathsf{Circumscription}}
\newcommand{\numdotcoNP}{\sf \#\cdot \mathsf{coNP}}
\newcommand{\numNP}{\sf \#\mathsf{NP}}

\begin{document}

\maketitle

\begin{abstract}
We investigate the computational complexity of the problem of counting the locally maximal
satisfying assignments of a Constraint Satisfaction Problem (CSP) over the Boolean domain $\{0,1\}$. 
A~satisfying assignment is \emph{locally maximal} if 
any new assignment which is obtained from it by changing a~$0$ to a~$1$
is  unsatisfying. 
For each constraint language~$\Gamma$,
$\numMaxCSP{\Gamma}$
denotes the problem of counting the locally maximal satisfying assignments, 
given an input CSP with constraints in~$\Gamma$.
We give a complexity dichotomy for the problem of \emph{exactly} counting the
locally maximal satisfying assignments and a complexity 
trichotomy for the problem of \emph{approximately} counting them.
Relative to the 
problem $\numCSP{\Gamma}$, which is the problem of counting \emph{all} 
satisfying assignments,  
the 
locally
maximal version can sometimes be easier but never harder.
This finding contrasts with the recent discovery that
approximately counting 
locally
maximal independent sets in a bipartite graph
is harder (under the usual complexity-theoretic assumptions)  than counting 
all independent sets. 
\end{abstract}
 

\section{Introduction}
 
A Boolean Constraint Satisfaction Problem (CSP) 
is a generalised satisfiability problem.  An instance of a Boolean 
CSP is a set of variables together with a collection of constraints that enforce certain
relationships between the variables.  These constraints are chosen from an agreed
finite set (``language'') $\Gamma$ of relations of various arities on 
the Boolean domain $\{0,1\}$.
The study of the computational complexity of Boolean CSPs has a long history, 
starting with Schaefer, who described the complexity of the basic decision problem:
is a given Boolean CSP instance satisfiable?  The computational complexity of the
satisfiability problem depends, of course, on the constraint language~$\Gamma$, becoming 
potentially harder as $\Gamma$ becomes larger and more expressive.
Schaefer showed~\cite{Schaefer} that, depending on $\Gamma$, the satisfiability 
problem is either polynomial-time solvable or NP-complete, and he provided a precise
characterisation of this dichotomy.  

The problem $\numCSP{\Gamma}$ is the 
problem of determining the number of satisfying assignments of
a CSP instance with constraint language~$\Gamma$.
A dichotomy for this counting problem was given by Creignou and Hermann~\cite{CH}.
\begin{theorem}(Creignou, Hermann~\cite{CH}) \label{thm:CH}
Let~$\Gamma$ be a constraint language with domain~$\{0,1\}$. The problem $\numCSP{\Gamma}$ is
in~$\FP$ if every relation in~$\Gamma$ is affine. Otherwise, $\numCSP{\Gamma}$ is $\numP$-complete.
\end{theorem}
A relation is \emph{affine} if it is expressible as the set of solutions to a system of
linear equations over the two-element field $\mathbb{F}_2$.
The constraint language~$\Gamma$ is said to be affine if and only if every constraint in~$\Gamma$
is affine. Thus,
Theorem~\ref{thm:CH} shows (assuming $\FP\neq \numP$) that $\numCSP{\Gamma}$ is tractable
if and only if the set of satisfying assignments can be expressed as the set of solutions
to a system of linear equations.

Since most constraint languages~$\Gamma$ lead to intractable counting problems,
it is natural to consider the complexity of \emph{approximately} counting
the satisfying assignments of a CSP.
Dyer, Goldberg, Greenhill and Jerrum~\cite{DGGJ} used approximation-preserving reductions
(AP-reductions)
between counting problems to explore the complexity of approximately computing solutions.
They identified three equivalence classes of interreducible counting 
problems within $\numP$:
(i)~problems that have a polynomial-time approximation algorithm or ``FPRAS'',
(ii)~problems that are equivalent to $\BIS$ under AP-reductions,
and (iii)~problems that are equivalent to $\SAT$ under AP-reductions.
Here, $\BIS$ is the problem of counting  independent sets in a bipartite
graph and $\SAT$ is the problem of counting the satisfying assignments of a Boolean 
formula in CNF{}.
Dyer, Goldberg and Jerrum~\cite{trichotomy} show that all Boolean counting CSPs 
can be classified using these classes.   
\begin{theorem} (\cite[Theorem 3]{trichotomy})
\label{thm:trichotomy}
Let $\Gamma$ be a constraint language with domain $\{0,1\}$. If every relation in $\Gamma$ is affine then
$\numCSP{\Gamma}$ is in $\FP$. Otherwise if every relation in~$\Gamma$ is in~$\IMtwo$  
then $\numCSP{\Gamma} \APeq \BIS$.
Otherwise $\numCSP{\Gamma} \APeq \SAT$.
\end{theorem}
In the statement of Theorem~\ref{thm:trichotomy}, $\APeq$ is the equivalence relation 
``interreducible via approx\-imation-preserving reductions''.  
In order to define $\IMtwo$, we must first define the binary implication relation
$\Implies=\{(0,0),(0,1),(1,1)\}$.
Then $\IMtwo$ is the set of  relations that can be 
expressed using conjunctions of these implications, together with  unary constraints. 
The precise definition of $\IMtwo$  
is given in Section~\ref{sec:notation},
along with precise definitions of the other concepts that appear in this introduction.

There are many other questions that one can ask about Boolean CSPs aside from 
deciding satisfiability and counting the satisfying assignments.
Here, we study the complexity of 
counting and 
approximately counting the number of locally maximal satisfying
assignments of
a CSP instance.  
A satisfying assignment is \emph{locally maximal} if any new assignment which is obtained
from it by changing a
single~$0$ to a~$1$ is unsatisfying. 
So 
local maximality
is with respect to the set of $1$'s in the satisfying assignment.
Also, 
it is with respect to local changes --- changing a single~$0$ to a~$1$.
Other notions of maximality are discussed in Section~\ref{sec:other}. 
 
Goldberg, Gysel and Lapinskas~\cite{MaximalBIS}
show (assuming that $\BIS$ 
is not equivalent to $\SAT$ under AP reductions) that 
counting 
locally maximal structures can be harder than counting
all structures.
In particular, Theorem~1 of~\cite{MaximalBIS} shows that counting 
the 
locally maximal independent sets in
a bipartite graph is equivalent to $\SAT$ under AP-reductions.
Obviously, counting \emph{all} independent
sets in a bipartite graph is exactly the problem $\BIS$, which is presumed
to be easier.
Thus, Goldberg, Gysel and Lapinskas have 
found an example of a (restricted) Boolean counting CSP (namely, $\BIS$)
where approximately counting 
locally maximal satisfying assignments is
apparently harder than approximately counting all satisfying assignments.
However, $\BIS$ isn't exactly a Boolean counting CSP --- rather, it is a
Boolean counting CSP with a restriction (bipartiteness) on the problem instance.
This prompted us  to investigate the complexity of 
approximating the number of satisfying assignments of unrestricted Boolean CSPs.
In particular, we study  $\numMaxCSP{\Gamma}$, the problem of counting locally maximal satisfying 
assignments of an instance of a Boolean CSP with constraint language~$\Gamma$.

Given the phenomenon displayed by $\BIS$, one might expect to find constraint 
languages~$\Gamma$ that exhibit a jump upwards in computational 
complexity when passing from $\numCSP{\Gamma}$ to $\numMaxCSP{\Gamma}$,
but we determine that this does not in fact occur.  The reverse may occur:  counting locally maximal 
vertex covers in a graph is trivial (there is just one), but counting all vertex 
covers is equivalent under AP-reducibility to~$\SAT$.  
It turns out that this trivial phenomenon,
which occurs when the property in question is monotone increasing,
is essentially
the only difference between $\numCSP{\Gamma}$ and $\numMaxCSP{\Gamma}$.

Our first result (Theorem~\ref{thm:exact}) presents a dichotomy
for the complexity of exactly solving $\numMaxCSP{\Gamma}$
for all Boolean constraint languages~$\Gamma$.
In most cases, $\numMaxCSP{\Gamma}$ is equivalent in complexity
to $\numCSP{\Gamma}$. However, if $\Gamma$ is \emph{essentially monotone} 
(the proper generalisation of the vertex cover property) 
then $\numMaxCSP{\Gamma}$ is in~$\FP$.
Our second result (Theorem~\ref{thm:main}) presents a trichotomy
for the complexity of \emph{approximately} solving  
  $\numMaxCSP{\Gamma}$. Once again, in most cases, 
$\numMaxCSP{\Gamma}$ is equivalent with respect to AP-reductions to
$\numCSP{\Gamma}$. 
The only exceptional case is the one that we have already seen ---
 if $\Gamma$ is \emph{essentially monotone} 
then $\numMaxCSP{\Gamma}$ is in~$\FP$.

The result leaves us with a paradox.  
There is a very direct reduction from
$\BIS$ to $\numCSP{\{\Implies\}}$ 
that is ``parsimonious'', i.e., 
preserves the number of solutions. 
So $\BIS$ is AP-reducible to  $\numCSP{\{\Implies\}}$
(which 
can also be seen from Theorem~\ref{thm:trichotomy}).
How can it be, then, that the complexity of approximately counting 
locally maximal 
independent sets in bipartite graphs
jumps upwards from the complexity of approximately counting all 
independent sets, 
whereas Theorem~\ref{thm:main} tells us that
$\numMaxCSP{\{\Implies\}}$ remains AP-equivalent to $\numCSP{\{\Implies\}}$?

The resolution of the paradox is as follows.  Suppose $G=(U,V,E)$
is an instance of $\BIS$, i.e., a graph with bipartition $U\cup V$
and edge set $E\subseteq U\times V$.
The parsimonious reduction
from $\BIS$ to $\numCSP{\{\Implies\}}$ simply interprets the vertices 
$U\cup V$ of instance $G$ as Boolean variables, and each edge 
$(u,v)\in E$ as a constraint $\Implies(u,v)$.  
Then there is an obvious bijection between independent sets in $G$ and 
satisfying assignments of the constructed instance of $\numCSP{\{\Implies\}}$,
but it involves interpreting 0 and 1 in different ways on the opposite 
sides of the bipartition:  on~$U$, $1$ means ``in the independent set'' 
while on~$V$, $1$ means ``out of the independent set''.  Of course, 
local maximality is not preserved by this change in interpretation, 
as it presumes a particular ordering on 0 and~1.  
Local maximality is sensitive to the precise encoding of solutions,
and parsimonious reductions are no longer enough to capture 
complexity equivalences. 
 
 Thus, we are left with the situation that
 approximately counting 
 locally
 maximal  independent sets in bipartite graphs is apparently more
 difficult than approximately counting all independent sets, but
 within the realm of Boolean constraint satisfaction, this phenomenon does not occur.
 Relative to $\numCSP{\Gamma}$,
 the problem $\numMaxCSP{\Gamma}$ can sometimes be easier
 but never harder.
 
\section{Notation and Preliminaries}\label{sec:notation}
  
\subsection{Locally maximal Constraint Satisfaction Problems} 
\label{sec:def}
  
A Boolean constraint language~$\Gamma$ is a set of relations on $\{0,1\}$.
Once we have fixed the constraint language $\Gamma$,
an \emph{instance}~$I$ of the CSP consists of a set~$V$ of \emph{variables\/}
and a set~$C$ of \emph{constraints}. Each
constraint has a \emph{scope,} which is a tuple of variables  
and a relation from~$\Gamma$ of the same
arity, which constrains the variables in the scope. An
\emph{assignment} $\sigma$ is a function from~$V$
to the Boolean domain~$\{0,1\}$. The
assignment~$\sigma$ is \emph{satisfying} if the scope of every
constraint is mapped to a tuple that is in the corresponding
relation. 
Given an assignment~$\sigma$, a variable~$v$, and a Boolean value $s$,
let $\tospin{\sigma}{v}{s}$  be the 
assignment defined as follows:  
$\tospin{\sigma}{v}{s}(v)=s$ and
for all $w\in V\setminus\{v\}$, $ \tospin{\sigma}{v}{s}(w)=\sigma(w)$.
(Thus, $\tospin{\sigma}{v}{s}$ agrees with~$\sigma$, except possibly at variable~$v$, which
it assigned Boolean value~$s$.)
We say that a satisfying assignment $\sigma$ is \emph{maximal} for~$v$ if either
\begin{itemize}
\item $\sigma(v)=1$, or
\item $\tospin{\sigma}{v}{1}$ is unsatisfying.
\end{itemize}
We say that the satisfying assignment $\sigma$ is \emph{locally maximal} if it is maximal for every variable $v\in V$.

For example, consider the ternary relation $R = \{(0,0,0), (0,0,1), (1,0,0), (0,1,1), (1,1,1)\}$
which excludes the three tuples $(0,1,0)$, $(1,1,0)$ and $(1,0,1)$. Let $\Gamma$ be the
size-one constraint language $\Gamma= \{R\}$.
Let $I$ be the instance with variable set $V = \{v_1,v_2, v_3,v_4,v_5\}$
and constraint set $C = \{R(v_1,v_2,v_3), R(v_3,v_4,v_5)\}$.  
Consider the following assignments.
$$
\begin{array}{|c|  ccccc |}
\hline
\sigma & \sigma(v_1) & \sigma(v_2) & \sigma(v_3) & \sigma(v_4) & \sigma(v_5)  \\
\hline 
\sigma_1 & 0 & 0 & 1 & 1 & 0\\
\sigma_2 & 0 & 0 & 1 & 1 & 1\\
\sigma_3 & 1 & 1 & 1 & 0 & 0\\
\hline
\end{array}
$$
The assignment $\sigma_1$ is not satisfying because the constraint 
$R(v_3,v_4,v_5)$ is not satisfied since $(1,1,0)$ is not in~$R$.
Assignments~$\sigma_2$ and~$\sigma_3$ are satisfying. 
Assignment~$\sigma_2$ 
is not maximal for $v_2$ since 
$\sigma_2(v_2)=0$ and
$\tospin{(\sigma_2)}{v_2}{1}$ is satisfying.
However, $\sigma_2$ is maximal for every other variable~$v_i$.
Assignment~$\sigma_3$ is locally maximal.

Given an instance~$I$ of a CSP with
constraint language $\Gamma$, the \emph{decision problem}
$\CSP{\Gamma}$  is to determine whether any assignment
satisfies~$I$. The \emph{counting problem} $\numCSP{\Gamma}$  is to
determine the \emph{number} of  satisfying
assignments of~$I$.  Finally, the 
\emph{locally maximal counting problem} $\numMaxCSP{\Gamma}$
 is to determine the number of locally maximal satisfying assignments of~$I$.

\subsection{Boolean Relations}   
\label{sec:boolrel}
   
A Boolean relation~$R$ is said to be \emph{$0$-valid} if the all-zero tuple is in~$R$
and it is said to be \emph{$1$-valid} if the all-one tuple is in~$R$.
 For every positive integer~$k$ and every $i\in\{1,\ldots,k\}$, let $e_{i,k}$ be the $k$-ary tuple
with a one in position~$i$ and zeroes in the other positions.
We say that a $k$-ary relation~$R$ is \emph{monotone} if, for every
tuple $(s_1,\ldots,s_k)\in R$ and every $i$,
the tuple $(s_1,\ldots,s_k) \vee e_{i,k}$ is also in~$R$, 
where $\vee$ is the \emph{or} operator, applied
position-wise.

The set of \emph{zero positions} of $R$, written $\zeroes{R}$, 
is $\{ i \in \{1,\ldots,k\} \mid 
\forall (s_1,\ldots,s_k) \in R, s_i=0
\}$.
Of course, $\zeroes{R}$ may be the empty set.
We use $\nonzeroes{R}$ to denote the set containing all other positions,
so $\nonzeroes{R} = [k]\setminus \zeroes{R}$.
We use $R^*$ to denote the relation induced on positions in~$\nonzeroes{R}$.
 We say that $R$ is \emph{essentially monotone} if  $R^*$ is monotone.
 
For example, the relation 
 $R=\{(0,0,1),(0,1,1),(1,1,1)\}$
is not monotone because the tuple $(0,0,1)$ has a zero in the first position 
but $(0,0,1) \vee (1,0,0) = (1,0,1)$ which is not in~$R$.
The relation 
$R = \{(0,0,1,1),(0,1,1,1)\}$
has $\zeroes{R} = \{1\}$ because all tuples in~$R$ have a zero in the first position,
and $\nonzeroes{R} = \{2,3,4\}$.
The relation $R^*$ induced on positions in $\nonzeroes{R}$ is 
$R^* = \{(0,1,1),(1,1,1)\}$. $R^*$ is monotone,  
so $R$ is essentially monotone.  

The binary implication relation $\Implies$
is defined as
$\Implies=\{(0,0),\allowbreak(0,1),\allowbreak(1,1)\}$.
We will also consider two unary relations~$\Uzero$ and~$\Uone$.
$\Uzero$ is defined by $\Uzero=\{(0)\}$.
The constraint $\Uzero(x)$ is often called ``pinning~$x$ to~$0$''.
Similarly, $\Uone$ is defined by 
$\Uone=\{(1)\}$ and
the constraint $\Uone(x)$
 is called ``pinning~$x$ to~$1$''.
 
 A Boolean \emph{co-clone}
 \cite{plainbases}
  is a set of Boolean relations
 containing the equality relation $\{(0,0),(1,1)\}$ 
 and closed under  
 certain operations.
 For completeness, the operations are 
 finite Cartesian products, projections, and identifications of variables, but
 it will not be necessary to define these here.
 We will not  require any information about co-clones,
 apart from the fact that the set of all affine relations is a co-clone, and so is
 a certain set $\IMtwo$ which we we have already discussed, and will define
below.

Suppose that $\Gamma$ is a co-clone.
A subset $B$ of $\Gamma$ is said to be a 
  ``plain basis'' for~$\Gamma$ \cite[Definition 1]{plainbases}
 if and only if 
it is the case that every  constraint~$C$ over~$\Gamma$
is logically equivalent to a conjunction of
constraints over~$B$  using the same variables as~$C$.

 Creignou et al.~\cite{plainbases} 
 have shown that 
 $\{\Implies,\Uzero,\Uone\}$ is a plain basis for the set~$\IMtwo$.
 In fact, $\IMtwo$ can just be defined  
 this way. A relation $R(x_1,\ldots,x_k)$ is in $\IMtwo$
 if and only if it is logically equivalent to
 a conjunction of constraints over 
 $\{\Implies,\Uzero,\Uone\}$ using the variables $x_1,\ldots,x_k$.
For example, consider the relation $R = \{(0,0,0,1),(0,1,1,1)\}$.
The relation~$R$ is in~$\IMtwo$ because
the constraint $R(w,x,y,z)$ is logically 
equivalent to the conjunction of
constraints 
$\Uzero(w)$, $\Implies(x,y)$, $\Implies(y,x)$ and $\Uone(z)$.
Note that the conjunction does not use any variables other than~$w$, $x$, $y$ and~$z$,
and this is important.

Let $\mathbb{L}$ denote the set of relations corresponding to
linear equations  over the two-element 
field $\mathbb{F}_2$.   
For example, the equality relation $\{(0,0),(1,1)\}$
is in~$\mathbb{L}$ because it corresponds to the equation $x_1\oplus x_2 = 0$.
Also,  the relation $\{(1,0,0),(0,1,0),(0,0,1),(1,1,1)\}$ 
is in~$\mathbb{L}$ because it corresponds to the equation $x_1\oplus x_2 \oplus x_3 = 1$.
Let $\mathbb{L}_k$ denote the set containing all relations in~$\mathbb{L}$
of arity at most~$k$.
We have already said that  a relation is \emph{affine} if it is expressible 
as the set of solutions to  a system of linear equations over  $\mathbb{F}_2$
and that a constraint language $\Gamma$ is affine iff every constraint in~$\Gamma$ is affine.
Creignou et al.~\cite{plainbases} note that $\mathbb{L}$ is a plain
basis for  
the set of affine relations. 
Thus,   if $\Gamma$ is affine, and every constraint in~$\Gamma$ has
arity at most~$k$, then 
every constraint~$C$ over~$\Gamma$ 
is logically equivalent to a conjunction of constraints over~$\mathbb{L}_k$ using the
same variables as~$C$.

\subsection{The complexity of approximate counting}
\label{sec:ap}

We now recall the necessary complexity-theoretic background from~\cite{DGGJ}.
A \emph{randomised approximation scheme\/} is an algorithm for
approximately computing the value of a
function~$f:\Sigma^*\rightarrow
\mathbb{N}$. The
approximation scheme has a parameter~$\varepsilon>0$
which specifies
the error tolerance.
A \emph{randomised approximation scheme\/} for~$f$ is a
randomised algorithm that takes as input an instance $ x\in
\alphabet^{\ast }$ (e.g., an encoding of a CSP
instance) and an error
tolerance $\varepsilon >0$, and outputs an
integer $z$
(a random variable on the ``coin tosses'' made by the algorithm)
such that, for every instance~$x$,
\begin{equation}
\label{eq:3:FPRASerrorprob}
\Pr \big[e^{-\epsilon} f(x)\leq z \leq
e^\epsilon f(x)\big]\geq \frac{3}{4}\, .
\end{equation}
The randomised approximation scheme is said to be a
\emph{fully polynomial randomised approximation scheme},
or \emph{FPRAS},
if it runs in time bounded by a polynomial
in $ |x| $ and $ \epsilon^{-1} $.
(See Mitzenmacher and Upfal~\cite[Definition 10.2]{MU05}.)

Suppose that $f$ and $g$ are functions from
$\alphabet^{\ast }$ to~$\mathbb{N}$. An
``approximation-preserving reduction'' (AP-reduction)
from $f$ to~$g$ is a randomised algorithm~$\mathcal{A}$ for
computing~$f$ using an oracle for~$g$.\footnote{The
reader who is not familiar with oracle Turing machines
can just think of this as an imaginary (unwritten)
subroutine for computing~$g$.}
The algorithm~$\mathcal{A}$ takes
as input a pair $(x,\varepsilon)\in\alphabet^*\times(0,1)$, and
satisfies the following three conditions: (i)~every oracle call made
by~$\mathcal{A}$ is of the form $(w,\delta)$, where
$w\in\alphabet^*$ is an instance of~$g$, and $0<\delta<1$ is an
error bound satisfying $\delta^{-1}\leq\poly(|x|,
\varepsilon^{-1})$; (ii) the algorithm~$\mathcal{A}$ meets the
specification for being a randomised approximation scheme for~$f$
(as described above) whenever the oracle meets the specification for
being a randomised approximation scheme for~$g$; and (iii)~the
run-time of~$\mathcal{A}$ is polynomial in $|x|$ and
$\varepsilon^{-1}$.    The key property of this notion of reducibility 
is that the class of functions computable by an FPRAS is closed 
under AP-reducibility.

If there is an AP-reduction from $f$ to $g$ then we say $f$ is \emph{AP-reducible
to~$g$} and write $f\APred g$.  If $f\APred g$ and $g\APred f$ then we say that 
$f$ and~$g$ are \emph{AP-interreducible} and write $f\APeq g$.
A class of counting problems that are all AP-interreducible 
has the property that either all problems in the class have an FPRAS 
or none do.  
A word of warning about terminology: the notation $\APred$ has also 
been used (see e.g.~\cite{crescenzi-badap}) to denote a different type of approximation-preserving reduction which applies to optimisation problems. We will not study optimisation problems in this paper, so hopefully this will not cause confusion.

The class of problems AP-interreducible with $\BIS$,
the problem of counting independent sets in a bipartite graph, has 
received particular attention.  It is generally believed that problems in
this class do not have an FPRAS.

\section{Our main results}

Our main results give a dichotomy for 
exactly solving $\numMaxCSP{\Gamma}$ and
a trichotomy for its approximation.
\begin{theorem}
\label{thm:exact}
Let $\Gamma$ be a constraint language with domain $\{0,1\}$. 
\begin{itemize}
\item 
If every relation in $\Gamma$ is essentially monotone then
$\numMaxCSP{\Gamma}$ is in $\FP$.
\item
If every relation in $\Gamma$ is affine then
$\numMaxCSP{\Gamma}$ is in $\FP$. 
\item Otherwise, $\numMaxCSP{\Gamma}$ is $\numP$-complete.
\end{itemize}

\end{theorem}

\begin{theorem}  
\label{thm:main}
Let $\Gamma$ be a constraint language with domain $\{0,1\}$. 
\begin{itemize}
\item 
If every relation in $\Gamma$ is essentially monotone then
$\numMaxCSP{\Gamma}$ is in $\FP$.
\item
If every relation in $\Gamma$ is affine then
$\numMaxCSP{\Gamma}$ is in $\FP$. 
\item
If $\Gamma$ has a relation that is not essentially monotone
and $\Gamma$ has a relation that is not affine,
but  every relation in~$\Gamma$ is in~$\IMtwo$  
then $\numMaxCSP{\Gamma} \APeq \BIS$.
\item
If $\Gamma$ has a relation that is not  essentially monotone and
a relation that is not affine and a relation that is not in~$\IMtwo$ then
$\numMaxCSP{\Gamma} \APeq \SAT$.
\end{itemize}
\end{theorem}   

Theorem~\ref{thm:exact} shows that if 
$\Gamma$ contains a relation~$R$ that is not essentially monotone and a relation~$R'$
that is not affine then $\numMaxCSP{\Gamma}$ is $\numP$-complete.
It is not necessary for~$R$ and~$R'$ to be distinct.
Similarly, the relations  in~$\Gamma$ witnessing 
``not essentially  monotone'', non-affineness  and   
non-contain\-ment in~$\IMtwo$ will not in general be distinct.

\section{Other notions of maximality}
\label{sec:other}

In Section~\ref{sec:def} we defined \emph{local maximality}.
A satisfying assignment~$\sigma$ is 
locally maximal if, 
for all $v$ with $\sigma(v)=0$, the configuration
$\tospin{\sigma}{v}{1}$ 
which is obtained from~$v$ by locally flipping 
the value of~$\sigma(v)$ from~$0$ to~$1$ is unsatisfying.

The study of approximately counting locally-optimal structures is motivated by the following
often-arising situation which is associated, for example, with 
Johnson, Papadimitriou and Yannakakis's
complexity class PLS (polynomial-time local search) \cite{JYPPLS}.
Often, it is easy to construct an arbitrary structure, difficult to construct a globally-optimal
structure, and of intermediate complexity to construct a locally-optimal structure.
A similar phenomenon arises in the study of \emph{listing} combinatorial structures
(see \cite{MaximalBIS} for details). 
Notably, \cite{MaximalBIS} found that this situation is not replicated in the context
of approximately counting independent sets in bipartite graphs --- a context in which approximately 
counting locally-optimal structures is (subject to complexity-theoretic assumptions) \emph{more
difficult} than counting globally-optimal structures.
This was the motivation for the present paper, which studies the problem of 
approximately counting locally-optimal 
(locally maximal) structures
in the context of CSPs. 
 
In this section, we  note that there are also other definitions of ``maximal'' that are not related to 
local optimality. One such example arises in the work of Durand, Hermann and Kolaitis~\cite{DHK}.
Following them, we describe their work in terms of \emph{minimality} rather than \emph{maximality}, but
this is not an essential difference.
The essential difference is that their notion of minimality  is 
based on subset inclusion rather than on local changes.
They  say that  a satisfying assignment $\sigma$ of a Boolean formula  is \emph{minimal}
if there is no satisfying assignment $\sigma'$ derived
from~$\sigma$ by  flipping any non-empty \emph{set} of values in the assignment 
from~$1$ to~$0$. 
The analogous version of our definition would only allow one value to flip.
They studied
 the problem $\circum$, in which the input is a Boolean formula, and the output
is the number of minimal satisfying assignments.
Assuming that the counting hierarchy does not collapse, this problem is not even in $\numP$.
Indeed, they show \cite[Theorem 5.1]{DHK} that it is complete in 
Hemaspaandra and Vollmer's complexity class
$\numdotcoNP$ \cite{HV}, which is equivalent to Valiant's class $\numNP$
\cite{Val1,Val2}.\footnote{A theorem of Toda and Watanabe~\cite{TodaWat} 
tells us that $\numNP$ is the same as $\numP$
once we close with respect to polynomial-time Turing reductions, but
Durand et al.\ study  higher counting classes via subtractive reductions under which the
higher counting classes are closed.} In \cite{DH}, Durant and Hermann
investigate the complexity of $\circum$ when the input Boolean formula has a prescribed form.
Their result contrasts   sharply with Theorem~\ref{thm:exact}.
For example, they show \cite[Theorem 4]{DH} that $\circum$ is $\numP$-complete when the
formula is restricted to be affine, even though affine is one of the easy cases in Theorem~\ref{thm:exact}.
This difference illustrates how
different 
subset-inclusion maximality and local maximality
really are. 
 
Counting 
globally optimal (\emph{maximum}) structures is different 
from counting 
either type of
\emph{maximal} structures and there is also some interesting
work about the former \cite{PB,HP}.

\section{Proofs}

We now give the proofs of our main theorems, Theorem~\ref{thm:exact} and Theorem~\ref{thm:main}.
We start with Theorem~\ref{thm:exact}.
Given a constraint language~$\Gamma$, it is easy to see that
$\numMaxCSP{\Gamma} \in \numP$ --- this is witnessed by the brute-force algorithm
which checks every assignment to the CSP instance and checks whether it is
a locally maximal satisfying assignment.
(Note that this check can be done in polynomial time.)
Theorem~\ref{thm:exact} follows from this fact and from 
Lemmas~\ref{lem:mon}, \ref{lem:affine} and \ref{lem:numphard},
which we will prove in the remainder of the paper.

\newcommand{\statelemmon}{Let $\Gamma$ be a constraint language with domain $\{0,1\}$. 
If every relation in $\Gamma$ is  essentially monotone then
$\numMaxCSP{\Gamma}$ is in $\FP$.}
\begin{lemma}
\label{lem:mon}
\statelemmon{}
\end{lemma}
   
\newcommand{\statelemaffine}{Let $\Gamma$ be a constraint language with domain $\{0,1\}$. 
If every relation in $\Gamma$ is affine then
$\numMaxCSP{\Gamma}$ is in $\FP$.}   
\begin{lemma}
\label{lem:affine}
\statelemaffine{}
\end{lemma} 

\newcommand{\statelemnumphard}{Let $\Gamma$ be a constraint language with domain $\{0,1\}$. 
If $\Gamma$ has a relation that is not affine and a relation that is not essentially monotone
then $  \numMaxCSP{\Gamma}$ is $\numP$-hard.}\begin{lemma}
\label{lem:numphard}
\statelemnumphard{}
\end{lemma}

 All problems in $\numP$ are AP-reducible to~$\SAT$ \cite{DGGJ}
so $\numMaxCSP{\Gamma} \APred \SAT$.   
Theorem~\ref{thm:main} follows from 
 this fact, 
 from Lemmas~\ref{lem:mon} and \ref{lem:affine}
 and from the following additional lemmas, which we will also prove 
in the remainder of the paper.

\newcommand{\statelembiseasy}{Let $\Gamma$ be a constraint language with domain $\{0,1\}$. 
If every relation in~$\Gamma$ is in~$\IMtwo$  
then $\numMaxCSP{\Gamma} \APred \BIS$.}   
\begin{lemma}
\label{lem:biseasy}
\statelembiseasy{}
\end{lemma}   

\newcommand{\statelembishard}{Let $\Gamma$ be a constraint language with domain $\{0,1\}$. 
If $\Gamma$ has a relation that is not affine and a relation that is not  essentially monotone
then $\BIS \APred \numMaxCSP{\Gamma}$.}
\begin{lemma}
\label{lem:bishard}
\statelembishard{}
\end{lemma}

\newcommand{\statelemsathard}{Let $\Gamma$ be a constraint language with domain $\{0,1\}$. 
If $\Gamma$ has a relation that is not affine and a relation that is not essentially monotone
and a relation that is not in $\IMtwo$
then $\SAT \APred \numMaxCSP{\Gamma}$.}
\begin{lemma}
\label{lem:sathard}
\statelemsathard{}
\end{lemma}

\subsection{Proofs of the easiness lemmas.
}

We start by proving Lemmas~\ref{lem:mon}, \ref{lem:affine} and
\ref{lem:biseasy}.

{\renewcommand{\thetheorem}{\ref{lem:mon}}
\begin{lemma} \statelemmon{}\end{lemma}}
 
 \begin{proof}
Consider a CSP instance~$I$ with a set~$V$ of variables and a set~$C$ of
constraints.
Let $U$ be the set of variables that are ``pinned'' to~$0$ by zero positions of~$R$.
Specifically,
$$U = \big\{v \in V \mid
\text{ there is a constraint
$R(v_1,\ldots,v_k)$ in $C$   such that 
$v_i=v$ and $i\in \zeroes{R}$
}\big\}.$$
Let $W = V \setminus U$.
We claim that every locally maximal satisfying assignment of~$I$
maps the variables in~$U$ to the Boolean value~$0$
and the variables in~$W$ to the Boolean value~$1$.
Thus, there is at most one locally maximal satisfying assignment, and it is
easy to check in polynomial time whether or not this
satisfying assignment exists.
Let us see why the claim is true.
It is clear from the definition of  $\zeroes{R}$ that
every satisfying assignment maps all 
variables in~$U$ to~$0$. 
But all of the induced constraints on variables in~$W$ are monotone
so if the instance~$I$ has a satisfying assignment
then the assignment that maps all variables in~$W$ to~$1$ is
the only locally maximal satisfying assignment.
 \end{proof}

{\renewcommand{\thetheorem}{\ref{lem:affine}}
\begin{lemma} \statelemaffine{}\end{lemma}} 
\begin{proof} 

Let $\Gamma$ be an affine constraint language with domain $\{0,1\}$.
Let $k$ be the maximum arity of any constraint in~$\Gamma$.
We know from Section~\ref{sec:boolrel}
that every constraint~$C$ over~$\Gamma$ is logically equivalent
to a conjunction of constraints over~$\mathbb{L}_k$ 
using the same variables as~$C$.
We can thus transform an instance $I$ of $\numMaxCSP{\Gamma}$ on variable set~$V$
into an equivalent instance $J$ of $\numMaxCSP{\mathbb{L}_k}$ on the same set of variables.  
Satisfying assignments of~$I$ correspond to satisfying assignments of~$J$
and locally maximal satisfying assignments of~$I$ correspond to locally maximal satisfying assignments of~$J$.

We will make one further transformation.
Let $W\subseteq V$ be the set of variables that are constrained in the  instance~$J$
(i.e., that occur in some constraint 
in~$J$). Let 
$U = V \setminus W$ be the set of unconstrained variables.
Let $J'$ be the instance on variable set~$W$
obtained from $J$ by removing the variables in~$U$.

Now, in any locally maximal satisfying assignment of~$J$, the variables 
$U$, being unconstrained, must take the value~1.  
Thus, there is a one-to-one correspondence between locally maximal satisfying assignments of~$J$
and locally maximal satisfying assignments of~$J'$.
 
Finally, 
observe that every satisfying assignment of $J'$ is locally maximal:
flipping any variable~$v$ from $0$ to~$1$ will violate all the constraints 
that involve $v$.  So the instance $J'$ has the same number of satisfying assignments
as locally maximal satisfying assignments.
Thus, we can count the locally maximal satisfying assignments
of~$I$ by counting the locally maximal satisfying assignments of~$J$ which
is the same as counting the locally maximal satisfying assignments of~$J'$
which is the same as counting all of the satisfying assignments of~$J'$.
This final step can be done by  Gaussian elimination.  
\end{proof}

Lemma~\ref{lem:biseasy}
follows directly from Lemma~\ref{lem:bisone} and~\ref{lem:bistwo}, which we prove next.

\begin{lemma}
\label{lem:bisone}
Let $\Gamma$ be a constraint language with domain $\{0,1\}$. 
If every relation in~$\Gamma$ is in~$\IMtwo$  
then $\numMaxCSP{\Gamma} \APred 
\numMaxCSP{\{\Implies\}}$.\end{lemma}

\begin{proof}
As we have noted in Section~\ref{sec:boolrel}, 
the  set $\IMtwo$ has the plain basis 
$B=\{\Implies,\Uzero,\Uone\}$. 
Thus, we can thus transform an instance   of $\numMaxCSP{\Gamma}$  
into an equivalent instance  of $\numMaxCSP{ B}$ on the same set of variables.  
So to finish, we just need to give an
AP-reduction from 
$\numMaxCSP{B}$ to $\numMaxCSP{\{\Implies\}}$.

Let $I$ be an instance of $\numMaxCSP{B}$ on variable set~$V$.
It is convenient to model the structure of~$I$  
as a directed graph $G(I)$ with vertex set~$V$: 
There is a  
directed edge from~$u$ to $v$ in~$G(I)$ whenever there
is a constraint $\Implies(u,v)$ in~$I$. 
Define subsets $V_1(I)$, $V_0(I)$ and $U(I)$ 
of~$V$ as follows.  Note that these sets can be computed from~$I$ in polynomial time.
\begin{itemize}
\item $v\in V_1(I)$ if, for some $u\in V$, there is a constraint $\Uone(u)$ in~$I$
and there is a directed path from~$u$ to~$v$ in~$G(I)$. Note that 
we include the empty directed path 
(of length~$0$)
so 
for every constraint $\Uone(u)$ in~$I$, the vertex~$u$ is in~$V_1(I)$.
\item $v\in V_0(I)$ if, for some $w \in V$, there is a constraint $\Uzero(w)$ in~$I$
and there is a directed path from~$v$ to~$w$ in~$G(I)$. Once again, we include length-$0$ paths.
\item $U(I) = V \setminus (V_0(I) \cup V_1(I))$.
\end{itemize}
 
The instance~$I$ is satisfiable if and only if $V_0(I) \cap V_1(I)$ is non-empty.
(To see this, note that if $V_0(I) \cap V_1(I)$ is non-empty, there is a variable~$u$
that is forced to take value~$0$ and value~$1$ in any satisfying assignment, which is impossible.
On the other hand, if the intersection is empty, then the assignment that
maps $V_0(I)$ to value~$0$ and the remaining vertices to value~$1$ is satisfying.)

If $I$ is unsatisfiable, then our AP-reduction just returns the number of satisfying assignments of~$I$
(which is zero) without using the oracle for $\numMaxCSP{\{\Implies\}}$.

Suppose instead that $I$ is satisfiable.
Let $I'$ be the instance of $\numMaxCSP{\{\Implies\}}$ on variable set~$U$
that is induced from~$I$.
That is, for any vertices $u_1$ and $u_2$ in~$U$,
$\Implies(u_1,u_2)$ is a constraint in~$I'$ if and only if it is a constraint in~$I$.
Then
the locally maximal satisfying assignments of~$I$ are in one-to-one correspondence with the locally maximal satisfying assignments of~$I'$.
\end{proof}
 
\begin{lemma}
\label{lem:bistwo}
$\numMaxCSP{\{\Implies\}} \APred \BIS$.
\end{lemma}
  
\begin{proof}
 
Let $I$ be an instance of $\numMaxCSP{\{\Implies\}}$ on variable set~$V$.
As in the proof of Lemma~\ref{lem:bisone},
It is convenient to model the structure of~$I$  
as a directed graph $G(I)$ with vertex set~$V$: 
There is a  
directed edge from~$u$ to $v$ in~$G(I)$ whenever there
is a constraint $\Implies(u,v)$ in~$I$. 
In the following, we refer to
the strongly-connected components of $G(I)$ as ``components''.

First, suppose that $|V|>1$ and that $G(I)$ has a singleton component $\{v\}$.
We will show below how to construct (in polynomial-time) an instance~$I'$ 
of $\numMaxCSP{\{\Implies\}}$ on  variable set  $V-\{v\}$
such that the number of locally maximal satisfying assignments of~$I$ is
equal to the number of locally maximal satisfying assignments of~$I'$.
 
Before giving the details of the construction,
we show how to use it to obtain the desired AP-reduction.
Given $I$, we repeat the construction as many times as necessary to obtain 
an instance $I^*$ of $\numMaxCSP{\{\Implies\}}$
such that 
$I^*$ has the same number of locally maximal satisfying assignments as~$I$
and either (1) $I^*$ has only one variable, or (2) $G(I^*)$ has no
singleton components. In Case~(1), the number of locally maximal satisfying assignments of~$I^*$
(and hence of~$I$) is one.  So consider Case~(2).
 Now note that
  every satisfying assignment of~$I^*$ 
is locally maximal since 
flipping the value of a single variable without flipping the rest of the variables
in its component does not preserve satisfiability. 
Thus, the number of locally maximal satisfying assignments of~$I$
is equal to the number of satisfying assignments of~$I^*$.
To finish, we use an oracle for $\BIS$ to approximately count
the satisfying assignments of~$I^*$.
This is possible   since  $\numCSP{\{\Implies\}}$ is AP-reducible 
  to $\BIS$
  by Theorem~\ref{thm:trichotomy} (which is from \cite{trichotomy}). 

To finish the proof, we give the construction.
So  suppose that $|V|>1$ and that $G(I)$ has a singleton component $\{v\}$.
Let $P$ be the (potentially empty) set of in-neighbours of vertex~$v$ in~$G(I)$
and let $S$ be the (potentially empty) set of out-neighbours of~$v$.
Construct $I'$ from~$I$ by deleting variable~$v$ and all constraints involving~$v$
and adding all constraints $\Implies(u,w)$ for $u\in P$ and $w\in S$.
To finish, we will give a bijection between the locally maximal satisfying assignments of~$I$ and~$I'$.

We start by partitioning the locally maximal satisfying assignments of~$I$
and $I'$ into three sets.
\begin{itemize}
\item Let $\Sigma_{1,*}$ be the set of locally maximal satisfying assignments $\sigma$ of~$I$ for which there
exists $u\in P$ with $\sigma(u)=1$. 
Let $\Sigma'_{1,*}$ be the set of locally maximal satisfying assignments $\sigma'$ of~$I'$ for which there
exists $u\in P$ with $\sigma'(u)=1$. 
We make the following deductions about every $\sigma \in \Sigma_{1,*}$ and $\sigma'\in \Sigma'_{1,*}$.
\begin{enumerate}[({A}1)]
\item $\sigma(v)=1$. (This follows since $\sigma$ is satisfying.)
\item    $\forall w\in S$, $\sigma(w)= \sigma'(w)=1$.
(This follows since $\sigma$ and $\sigma'$ are satisfying.)
\item  
For every $u\in P$ with $\sigma'(u)=0$,
there is an out-neighbour $z'$ of~$u$ in $G(I')$ which has $\sigma'(z')=0$.
For every $u\in P$ with $\sigma(u)=0$,
there is an out-neighbour $z$ of~$u$ in $G(I)$ which is not equal to~$v$ and has $\sigma(z)=0$.

These follow since $\sigma'$ and $\sigma$ are maximal for~$u$ and $\sigma(v)=1$.)
\end{enumerate}
\item Let $\Sigma_{*,0}$ be the set of locally maximal satisfying assignments $\sigma$ of~$I$ for which there
exists $w\in S$ with $\sigma(w)=0$.
Let $\Sigma'_{*,0}$ be the set of locally maximal satisfying assignments $\sigma'$ of~$I'$ for which there
exists $w\in S$ with $\sigma'(w)=0$.
We make the following deductions about every $\sigma \in \Sigma_{*,0}$
and $\sigma' \in \Sigma'_{*,0}$.
\begin{enumerate}[({B}1)]
\item $\sigma(v)=0$. (This follows since $\sigma$ is satisfying.)
\item   
$\forall u\in P$, $\sigma(u)= \sigma'(u)=0$. (This follows since $\sigma$ and $\sigma'$ are satisfying.)
 
\end{enumerate}
\item Let $\Sigma_*$ be the set of all other locally maximal satisfying assignments of~$I$.
Let $\Sigma'_*$ be the set of all other locally maximal satisfying assignments of~$I$.
We make the following deductions about every $\sigma \in \Sigma_{*}$ and $\sigma' \in \Sigma'_{*}$.
\begin{enumerate}[({C}1)]
\item For all $u\in P$, $\sigma(u)= \sigma'(u)=0$. (This follows from the definitions of $\Sigma_{1,*}$ and $\Sigma'_{1,*}$.)
\item For all $w\in S$, $\sigma(w)= \sigma'(w)=1$. (This follows from the definitions of $\Sigma_{*,0}$ and $\Sigma'_{*,0}$.)
\item $\sigma(v)=1$. (This follows because $\sigma$ is maximal for~$v$.)
 \item For every $u\in P$,
there is an out-neighbour $z'$ of~$u$ in $G(I')$ which  has $\sigma'(z')=0$. 
 For every $u\in P$,
there is an out-neighbour $z\neq v$ of~$u$ in $G(I)$ which  has $\sigma(z)=0$.
(This follows because $\sigma'$ and $\sigma$ are maximal for~$u$ and $\sigma(v)=1$.)
 \end{enumerate}
\end{itemize}

Given $\sigma\in \Sigma_{1,*}$, let $\sigma''$ be the induced assignment of~$I'$.
Since by (A2) $\sigma(w)=1$ for every $w\in S$,
all of the new constraints $(u,w)$ in $I'$ are satisfied, 
so $\sigma''$ is satisfying.
By construction, $\sigma''$ is maximal for the vertices outside of~$P$ (since vertices outside of~$P$
have the same out-neighbours in~$G(I)$ and~$G(I')$).
We now check that it is maximal for vertices $u\in P$. This follows from~(A3).
Thus, we have given an injection from $\Sigma_{1,*}$ into $\Sigma'_{1,*}$.
We now show that the reverse direction is an injection from $\Sigma'_{1,*}$ to $\Sigma_{1,*}$.
Consider $\sigma'\in \Sigma_{1,*}'$ and let $\sigma''$ be the assignment of~$I$
formed from~$\sigma'$ by taking $\sigma''(v)=1$ (satisfying (A1) for $\sigma=\sigma''$).
Since by (A2)
$\sigma'(w)=1$ for all $w\in S$, we conclude that $\sigma''$ 
satisfies all of the constraints involving~$v$ in~$I$ so $\sigma''$
is satisfying.
Once again, we must check that $\sigma''$ is maximal for vertices $u\in P$, and this follows
from (A3).

In a similar way, we will establish a bijection from $\Sigma_{*,0}$ to $\Sigma'_{*,0}$.
Given $\sigma \in \Sigma_{*,0}$, let $\sigma''$ be the induced assignment of~$I'$.
(B2)  allows us to conclude that $\sigma''$ is satisfying.
The edges in~$G(I')$ from~$P$ to~$w$ allow us to conclude that $\sigma''$ is maximal for all $u\in P$
so it is locally maximal. 
Going the other direction,
consider $\sigma'\in \Sigma_{*,0}$ and let $\sigma''$ be the assignment of~$I$
formed from~$\sigma'$ by taking $\sigma''(v)=0$ (satisfying (B1) for $\sigma=\sigma''$).
Since by (B2) $\sigma'(u)=0$ for all $u\in P$, $\sigma''$ is satisfying.
Since $\sigma''(u) =0$, $\sigma''$ is maximal for every vertex in~$P$.
Since $\sigma''(w)=0$, $\sigma''$ is maximal for~$v$. Thus, it is locally maximal.

In exactly the same way, we establish a bijection from $\Sigma_*$ to $\Sigma'_*$.
We conclude that  $I'$ has the same number of locally maximal satisfying assignments as~$I$,
 so we have completed the proof.
\end{proof}

\subsection{Maximality Gadgets}

We will now examine gadgets that are useful for proving  Lemmas~\ref{lem:numphard}, \ref{lem:bishard}
and~\ref{lem:sathard}.
We start with some useful definitions.

\begin{definition}\label{def:gadget}
Let $I$ be a CSP instance with 
a set~$V$ of  variables and a distinguished variable~$r\in V$.
We use $\maxzero{I}{r}$ to denote the number of locally maximal satisfying assignments~$\sigma$ of~$I$
with $\sigma(r)=0$.
We use $\maxone{I}{r}$ to denote the number of locally maximal satisfying assignments~$\sigma$ of~$I$
with $\sigma(r)=1$.
Finally, we use $\bad{I}{r}$ to denote the number of satisfying assignments~$\sigma$ of~$I$
such that $\sigma$ is maximal for every variable in~$V\setminus \{r\}$ but 
$\sigma$ is not maximal for~$r$.
A \emph{maximality gadget} for a relation~$R$ is a CSP instance~$I$ with constraint language~$\{R\}$
and distinguished variable~$r$
such that $\maxzero{I}{r} = \maxone{I}{r}=1$ and $\bad{I}{r}=0$.
\end{definition}

Definition~\ref{def:gadget} could be weakened since we do not really need
$\maxzero{I}{r}$ and $ \maxone{I}{r}$ to be~$1$, we only need them to be equal and non-zero.
However, since our constructions satisfy the stronger definition, we simplify the remainder
of the paper by using the stronger definition that we have stated.
 
Recall that $R^*$ is the relation induced from $R$ on 
the (non-zero) positions in $\nonzeroes{R}$. We first relate maximality 
gadgets for~$R$ and~$R^*$.

\begin{lemma} \label{lem:star}
Let $R$ be a Boolean relation.
If there is a maximality gadget for $R^*$ then there is a maximality 
gadget for~$R$.
\end{lemma}

\begin{proof}
Suppose $R$ has arity~$h$ and $R^*$ has arity $k\leq h$.  
Without loss of generality, suppose that that 
$\nonzeroes{R}=\{1,\ldots,k\}$ and $\zeroes{R}=\{k+1,\ldots,h\}$;
thus $R^*$ is the restriction of $R$ to the first $k$~places.
Suppose we have a   maximality gadget   for~$R^*$.  
That is, we have a CSP instance $I^*$, with $R^*$-constraints, 
on variables $V$, with a distinguished variable $r\in V$, 
satisfying Definition~\ref{def:gadget}. 
We show how to construct a maximality gadget  for~$R$.
 
If $k=h$ then $R^*=R$ and there is nothing to show, so suppose $k<h$.
Construct a new CSP instance~$I$ with $R$-constraints as follows.  
The variable set of $I$ is $V\cup\{w\}$, where $w$ is a 
new variable that is not in~$V$.
Replace each constraint $R^*(v_1,\ldots,v_k)$ in $I^*$ by a constraint 
$R(v_1,\ldots,v_k,w,\ldots,w)$ in~$I$.
(Note that there are $h-k$ occurrences of the variable~$w$.)
Since $k<h$, the variable~$w$ is forced to~$0$ in any satisfying assignment 
of~$I$.  Thus, there is a bijection between satisfying assignments of~$I^*$ 
and satisfying assignments of~$I$, obtained by setting $w$ to~$0$.  
The bijection preserves locally maximality.
From this bijection it is clear that $\maxzero{I}{r} = \maxzero{I^*}{r} = 1$,
$\maxone{I}{r}=\maxone{I^*}{r}=1$ and $\bad{I}{r}=\bad{I^*}{r}=0$.
Thus, $I$ is a maximality gadget for~$R$.
\end{proof}

Lemma~\ref{lem:gadget} below shows
 how maximality gadgets can be constructed.
Some of the constructions are  a little bit reminiscent of the ``strict, perfect, faithful implementations'' of Creignou,
Khanna, and Sudan~\cite{CKS} (see Lemmas 5.24 and 5.25, Claim 5.31 and Lemma 5.30).
In
Lemmas~13--15 of~\cite{trichotomy}, we use these implementations to
realise one of the Boolean functions~$\Implies$, $\NAND=\{(0,0),(0,1),(1,0)\}$ 
or $\OR=\{(0,1),(1,0),(1,1)\}$ using a relation~$R$.  However, there are two important 
differences.
(1) The implementations of~\cite{CKS} and~\cite{trichotomy} allow
the use of auxiliary variables, and the use of these must be carefully controlled in order to
satisfy the maximality constraints in Definition~\ref{def:gadget}.
(2) While implementations of~$\Implies$ and~$\NAND$ are useful for maximality gadgets,
implementations of~$\OR$ do not seem to be useful. Thus, 
the   implementations from earlier papers do not suffice for building maximality gadgets.

 \begin{lemma}
 \label{lem:gadget}
Let $R$ be a Boolean relation that is  not essentially monotone.
Then there is a maximality gadget for~$R$.
\end{lemma} 
\begin{proof} 
  
 Since $R$ is not essentially monotone, we know that $R^*$ is not monotone.   
 Suppose $R^*$ has arity $k$; clearly $k\geq1$ (in fact it is not difficult to 
 see that $k\geq2$). A simple but useful observation
 about $R^*$ is the following.
 \begin{equation}\label{eq:simpleobs}
 \text{For all $i$, $1\leq i\leq k$, there exists a tuple $(v_1,\ldots,v_k)\in R^*$ with $v_i=1$}.
 \end{equation}
 This follows from the fact that $i\notin\zeroes{R}$.
 
 We consider several cases, depending on~$R$.
 In each case, we construct a maximality gadget for~$R^*$, which by Lemma~\ref{lem:star} gives
 a maximality gadget for~$R$.

 \begin{enumerate}[{\bf {Case} 1}]
 
 \item \label{Case1} {\bf  $R^*$ is $0$-valid and $1$-valid.}
 Since $R^*$ is not monotone, it cannot be the complete relation.
 Let $s=(s_1,\ldots,s_k)$ be a tuple not in $R^*$.  Our maximality 
 gadget~$I$ has two variables $r$ and $x$, and two constraints $R^*(a_1,\ldots,a_k)$
 and $R^*(b_1,\ldots,b_k)$.  
 Each $a_j$ or~$b_j$ stands for an occurrence 
 of $r$ or~$x$; specifically,
 \begin{itemize}
 \item If $s_j=0$ then $a_j=r$ and $b_j=x$.
 \item If $s_j=1$ then $a_j=x$ and $b_j=r$.
 \end{itemize}
The following table lists the four possible assignments of~$(r,x)$
and analyses whether they are satisfying.
$$
\begin{array}{|c|cc| l |}
\hline
\sigma& \sigma(r) &  \sigma(x) &  \text{$I(r,x)$ satisfied?}    \\
\hline 
\sigma_1&0 & 0 & \mbox{yes, since $R^*$ is $0$-valid}\\
\sigma_2&0 & 1 & \mbox{no, since $\sigma(a_1,\ldots,a_k)=s\notin R^*$ }\\
\sigma_3&1 & 0 & \mbox{no, since $\sigma(b_1,\ldots,b_k)=s\notin R^*$ }\\
\sigma_4&1 & 1 & \mbox{yes, since $R^*$ is $1$-valid}\\
\hline
\end{array}
$$
Since $\sigma_1$ and $\sigma_4$ are the unique satisfying assignments $\sigma$
with $\sigma_1(r)=0$ and $\sigma_4(r)=1$, respectively, and both are locally maximal, we have
$\maxzero{I}{r}=\maxone{I}{r}=1$.  As there are no other satisfying assignments,
$\bad{I}{r}=0$.  So the conditions for $I$ to be 
a maximality gadget are satisfied.
   
 \item \label{Case2} {\bf   $R^*$ is $0$-valid but not $1$-valid.} 
 Let $m$ be the maximum number of  ones in any tuple in~$R^*$.
 By observation~(\ref{eq:simpleobs}), $m\geq1$, and since  
 $R^*$ is not $1$-valid, $m<k$. 
 Let $s=(s_1,\ldots,s_k)$ be a tuple in~$R^*$ with $m$ ones.
 Again by observation~(\ref{eq:simpleobs}), 
 there is a tuple $s'=(s'_1,\ldots,s'_k) \in R^*$
 such that, for some~$i$, 
 $s_i=0$ and $s'_i=1$.
 Note that the tuple $s''=s\vee s'$  
 is not in~$R^*$, since it has more than $m$ ones. 
 We split the analysis into  two sub-cases.
 
 \begin{enumerate} 
 \item {\bf The tuple $t=s\wedge\neg s'$ is in $R^*$.}
 (The tuple $t$ has $t_j=1$ precisely when $s_j=1$ and $s'_j=0$.)
 The gadget $I$ will have  variables $r$, $x$ and $w$ and constraints
 $R^*(w,\ldots,w)$ and
 $R^*(a_1,\ldots,a_k)$ where $a_j$ is defined 
 as follows. 
 
 \begin{itemize}
 \item If $s'_j=1$ then $a_j = x$.
 \item If $s_j'=0$ and $s_j=0$ then $a_j=w$.
 \item If $s_j'=0$ and $s_j=1$ then $a_j=r$.
 \end{itemize}
  
 Since $R^*$ is $0$-valid but not $1$-valid, the constraint $R^*(w,\ldots,w)$ 
 ensures that every satisfying assignment $\sigma$ of~$I$ has $\sigma(w)=0$.
 So we consider the four potential satisfying assignments.
 $$
\begin{array}{|c|  ccc| l |}
\hline
\sigma & \sigma(r) & \sigma(x) & \sigma(w) & \text{$I(r,x,w)$ satisfied?} \\
\hline 
\sigma_1&0 & 0 & 0 & \mbox{yes, since $R^*$ is $0$-valid}\\
\sigma_2&0 & 1 & 0 & \mbox{yes, since $\sigma(a_1,\ldots,a_k)=s'\in R^*$ }\\
\sigma_3&1 & 0 & 0 & \mbox{yes, since $\sigma(a_1,\ldots,a_k)=t \in R^*$ }\\
\sigma_4&1 & 1 & 0 & \mbox{no, since $\sigma(a_1,\ldots,a_k)=s''\notin R^*$}\\
\hline
\end{array}
$$
The assignment $\sigma_1$ is not maximal for~$x$
but $\sigma_2$ and $\sigma_3$, are locally maximal, so
$\maxzero{I}{r}=\maxone{I}{r}=1$ and $\bad{I}{r}=0$.  

\item {\bf The tuple $t=s\wedge\neg s'$ is not in $R^*$.}
The gadget $I$ will have  variables $r$, $x$ and $w$ and the constraints
$R^*(w,\ldots,w)$, $R^*(a_1,\ldots,a_k)$ and $R^*(b_1,\ldots,b_k)$,
 where $a_j$ and $b_j$ are defined as follows.
 \begin{itemize}
 \item If $s_j=0$ then $a_j = b_j=w$.
 \item If $s_j=1$ and $s_j'=0$ then $a_j=x$ and $b_j=r$.
 \item If $s_j=1$ and $s_j'=1$ then $a_j=r$ and $b_j=x$.
 \end{itemize}
 Since $R^*$ is $0$-valid but not $1$-valid,
 every satisfying assignment $\sigma$ has $\sigma(w)=0$.
 So we consider the four possible satisfying assignments.
 $$
\begin{array}{|c|  ccc| l |}
\hline
\sigma &  \sigma(r) & \sigma(x) & \sigma(w) & \text{$I(r,x,w)$ satisfied?} \\
\hline 
\sigma_1 & 0 & 0 & 0 & \mbox{yes, since $R^*$ is $0$-valid}\\
\sigma_2 & 0 & 1 & 0 & \mbox{no, since $\sigma(a_1,\ldots,a_k)=t\notin R^*$}\\
\sigma_3 & 1 & 0 & 0 & \mbox{no, since $\sigma(b_1,\ldots,b_k)=t\notin R^*$}\\
\sigma_4 & 1 & 1 & 0 & \mbox{yes, since $\sigma(a_1,\ldots,a_k)=\sigma(b_1,\ldots,b_k)=s\in R^*$}\\
\hline
\end{array}
$$
Both satisfying assignments are locally maximal so 
$\maxzero{I}{r}=\maxone{I}{r}=1$ and $\bad{I}{r}=0$.  
\end{enumerate}

 \item \label{Case3} {\bf  $R^*$ is not $0$-valid but is $1$-valid.}    
 As $R^*$ is not monotone, we may choose a tuple $s=(s_1,\ldots,s_k)$ in~$R^*$, 
 and an index $i\in [k]$
 such that $s_i=0$ and $s'=s\vee e_{i,k}$ is not in~$R^*$.
 The gadget~$I$ will have  variables $r$, $x$ and $y$ and the  
 constraints   
 $R^*(y,\ldots,y)$, $R^*(a_1,\ldots,a_k)$ and $R^*(b_1,\ldots,b_k)$ where 
 $a_i=x$ and $b_i=r$ and for $j\neq i$, $a_j$ and $b_j$ are defined as follows.
 \begin{itemize}
 \item If $s_j'=0$ then $a_j=r$ and $b_j=x$.
 \item If $s_j' =1$ then $a_j=b_j=y$.
 \end{itemize}
 Since $R^*$ is $1$-valid but not $0$-valid, the 
 constraint $R^*(y,\ldots,y)$ ensures that every 
 satisfying assignment~$\sigma$ of~$I$
 has $\sigma(y)=1$. 
 We consider the potential satisfying assignments with $\sigma(y)=1$.
 $$
\begin{array}{|c|  ccc| l |}
\hline
 \sigma & \sigma(r) & \sigma(x) & \sigma(y) &  \text{$I(r,x,y)$ satisfied?}  \\
\hline 
\sigma_1 & 0 & 0 & 1 & \mbox{yes, since $\sigma(a_1,\ldots,a_k)=\sigma(b_1,\ldots,b_k)=s\in R^*$}\\
\sigma_2 & 0 & 1 & 1 & \mbox{no, since $\sigma(a_1,\ldots,a_k)=s'\notin R^*$}\\
\sigma_3 & 1 & 0 & 1 & \mbox{no, since $\sigma(b_1,\ldots,b_k)=s'\notin R^*$}\\
\sigma_4 & 1 & 1 & 1 & \mbox{yes, since $R^*$ is $1$-valid}\\
\hline
\end{array}
$$ 
Note that $\sigma_1$ and $\sigma_4$ are locally maximal, so
$\maxzero{I}{r}=\maxone{I}{r}=1$ and $\bad{I}{r}=0$.

\item \label{Case4} {\bf   $R^*$ is not $0$-valid and not $1$-valid.}
We split the analysis into two sub-cases.

\begin{enumerate}
\item {\bf  There is a tuple $s\in R^*$ such that $\neg s\in R^*$.}
The gadget~$I$ will have  variables $r$ and~$x$ and the single 
constraint $R^*(a_1,\ldots,a_k)$ where 
\begin{itemize}
 \item If $s_j=0$ then $a_j=r$.
 \item If $s_j=1$ then $a_j=x$.
\end{itemize}
The potential satisfying assignments are 
 $$
\begin{array}{|c|  cc| l |}
\hline
\sigma & \sigma(r) & \sigma(x) & \text{$I(r,x)$ satisfied?}  \\
\hline 
\sigma_1 & 0 & 0 & \mbox{no , since $R^*$ is not $0$-valid}\\
\sigma_2 & 0 & 1 & \mbox{yes, since $\sigma(a_1,\ldots,a_k)=s\in R^*$  }\\
\sigma_3 & 1 & 0 & \mbox{yes since $\sigma(a_1,\ldots,a_k)=\neg s\in R^*$  }\\
\sigma_4 & 1 & 1 & \mbox{no, since $R^*$ is not $1$-valid}\\
\hline
\end{array}
$$
Note that $\sigma_2$ and $\sigma_3$ are locally maximal, so
$\maxzero{I}{r}=\maxone{I}{r}=1$ and $\bad{I}{r}=0$.

\item{\bf There is no tuple $s$ with $s\in R^*$ and $\neg s\in R^*$.}
Let $m$ be the maximum number of ones in any tuple in~$R^*$.
By observation~(\ref{eq:simpleobs}) $m\geq1$, and since $R^*$ is not $1$-valid, $m<k$.
Let $s=(s_1,\ldots,s_k)$ be a tuple in~$R^*$ with $m$ ones.
Again by observation~(\ref{eq:simpleobs}) there is a tuple $s'=(s'_1,\ldots,s'_k) \in R^*$
such that, for some $i$,
$s'_i=1$   and $s_i=0$.
Note that the tuple $s\vee s'$ is not in~$R^*$, since it has more than $m$~ones. 
The gadget $I$ will have four  variables $r$, $x$, $y$ and $w$ and the constraints
$R^*(a_1,\ldots,a_k)$ and $R^*(b_1,\ldots,b_k)$,
where $a_j$ and $b_j$ are defined as follows.
 \begin{itemize}
 \item If $s_j=s'_j=0$ then $a_j=b_j=w$.
 \item If $s_j=0$ and $s'_j=1$ then $a_j=w$ and $b_j=x$.
 \item If $s_j=1$ and $s'_j=0$ then $a_j=y$ and $b_j=r$.
 \item If $s_j=s'_j=1$ then $a_j=b_j=y$.
 \end{itemize} 
Since $R^*$ is not $0$-valid or $1$-valid, and
$s\in R^*$ but $\neg s\notin R^*$,
the constraint  $R^*(a_1,\ldots,a_k)$
ensures that every satisfying assignment~$\sigma$ has
$\sigma(w)=0$ and $\sigma(y)=1$.
So we consider the four possible satisfying assignments.
$$
\begin{array}{|c|  cccc| l |}
\hline
\sigma & \sigma(r) & \sigma(x) & \sigma(y) & \sigma(w) &  \text{$I(r,x,y,w)$ satisfied?}  \\
\hline 
\sigma_1 & 0 & 0 & 1 & 0 & \mbox{maybe}\\
\sigma_2 & 0 & 1 & 1 & 0 & \mbox{yes, since $\sigma(b_1,\ldots,b_k)=s'\in R^*$ }\\
\sigma_3 & 1 & 0 & 1 & 0 & \mbox{yes, since $\sigma(b_1,\ldots,b_k)=s\in R^*$ }\\
\sigma_4 & 1 & 1 & 1 & 0 & \mbox{no, since $\sigma(b_1,\ldots,b_k)=s\vee s'\notin R^*$ }\\
\hline
\end{array}
$$
There is no need to determine whether $\sigma_1$ is satisfying.
If it is, then it is not maximal for~$x$. The satisfying assignments
$\sigma_2$ and $\sigma_3$ are locally maximal.
So 
$\maxzero{I}{r}=\maxone{I}{r}=1$ and $\bad{I}{r}=0$.    
 \end{enumerate}
\end{enumerate}
\end{proof}
 
\subsection{Proofs of the hardness lemmas}
 
 We now prove Lemmas~\ref{lem:numphard}, \ref{lem:bishard}
and~\ref{lem:sathard}.
 
  {\renewcommand{\thetheorem}{\ref{lem:numphard}}
\begin{lemma} \statelemnumphard{}\end{lemma}}
\begin{proof}

Suppose that $R_1$ is an arity $k_1$ relation in $\Gamma$ that is not affine and
$R_2$ is an arity-$k_2$ relation in  $\Gamma$ that is not essentially monotone.
Let $k=k_1+k_2$
and let $R$ be the Cartesian product of~$R_1$ and~$R_2$.
Specifically, 
$$R = \{(x_1,\ldots,x_k) \mid (x_1,\ldots,x_{k_1})\in R_1, (x_{k_1+1},\ldots,x_k)\in R_2\}.$$
Since $R_1$ is not affine, neither is~$R$.
Thus, Theorem~\ref{thm:CH} (due to Creignou and Hermann)
shows that $\numCSP{\{R\}}$ is $\numP$-complete.
 
Since $R_2$ is not essentially monotone, $R$ is not essentially monotone.
Thus, we can use Lemma~\ref{lem:gadget} 
to obtain a maximality gadget $I$ for~$R$ with 
variable set~$V$ and
some distinguished variable~$r$.
We will  
next use the maximality gadget to give a 
polynomial-time Turing reduction from
$\numCSP{\{R\}}$ to  $\numMaxCSP{\{R\}}$.

First, consider the assignments of~$I$. The definition of maximality gadget ensures the following.
\begin{enumerate}
\item Since   $\bad{I}{r}=0$,
every satisfying assignment $\sigma$ of~$I$ that is maximal for every variable in $V\setminus\{r\}$ is also maximal for~$r$.
\item \label{OKone}
Since $\maxzero{I}{r}=1$ there is exactly one
satisfying assigment $\sigma_0$ of~$I$ that is maximal for every variable in $V\setminus\{r\}$ and satisfies $\sigma_0(r)=0$.
Note that $\sigma_0$ is maximal for~$r$.
\item \label{OKtwo} Since $\maxzero{I}{r}=1$ there is exactly one 
satisfying assignment $\sigma_1$ of~$I$ that is maximal for every variable in $V\setminus\{r\}$ and satisfies $\sigma_1(r)=1$.
Note that $\sigma_1$ is maximal for~$r$.
\end{enumerate} 

Now consider an instance~$J$ of $\numCSP{\{R\}}$ with vertex set $U$.
We will construct an instance~$J'$ of $\numMaxCSP{\{R\}}$
with $|V| \times |U|$ variables.
For every variable $u\in U$,
let $V_u$ be a set of $|V|$ variables consisting of
variable $u$ and $|V|-1$ new variables.
Let $I_u$ be a copy of the
maximality gadget~$I$ using the variables~$V_u$
with distinguished variable~$u$.
Finally, let $J'$ be the instance of $\numMaxCSP{\{R\}}$
with variable set $\bigcup_{u\in U} V_u$
and with all of the constraints in each of the instances~$I_u$
and with all of the further $R$-constraints inherited from~$J$ (these constraints inherited from~$J$
constrain the vertices in~$U$).

We will next show that the  
satisfying assignments of~$J$ are in one-to-one correspondence 
with locally maximal satisfying assignments of~$J'$.  
By construction, any locally maximal satisfying assignment of~$J'$
induces a satisfying assignment of~$J$ (just look at the induced assignment on variables in~$U$).

Consider any satisfying assignment $\sigma$ of~$J$.
Consider any variable $u\in U$.
If $\sigma(u)=0$ then by item~\eqref{OKone} above, there
is exactly one way to extend $\sigma$ to the vertices in $V_u$ that is maximal
for all variables in $V_u \setminus \{u\}$. This extension is also maximal for~$u$ itself.
(Thus, even if the constraints in~$J$ would allow the assignment at~$u$ to be flipped to a~$1$,
the unique extension of the assignment to~$V_u$ does not allow this.)
If $\sigma(u)=1$ then by item~\eqref{OKtwo},there
is exactly one way to extend $\sigma$ to the vertices in $V_u$ that is maximal
for all variables in $V_u \setminus \{u\}$. This extension is also maximal for~$u$ itself.
Thus, $\sigma$ can be extended in exactly one way to a locally maximal satisfying assignment of~$J'$.

So we have shown that  the  
satisfying assignments of~$J$ are in one-to-one correspondence 
with locally maximal satisfying assignments of~$J'$.  
Since $\numCSP{\{R\}}$ 
is $\numP$-hard, we have proved that $\numMaxCSP{\{R\}}$ is $\numP$-hard.

Finally, there is a trivial polynomial-time Turing reduction from
  $\numMaxCSP{\{R\}}$ to 
  $\numMaxCSP{\Gamma}$
since every instance of $\numMaxCSP{\{R\}}$ can be written as an instance of
$\numMaxCSP{\Gamma}$.
\end{proof}

   {\renewcommand{\thetheorem}{\ref{lem:bishard}}
\begin{lemma} \statelembishard{}\end{lemma}}
\begin{proof}

Even though we require an AP-reduction, the proof is essentially the same as the proof
of Lemma~\ref{lem:numphard}.
Suppose that $R_1$ is an arity $k_1$ relation in $\Gamma$ that is not affine and
$R_2$ is an arity-$k_2$ relation in  $\Gamma$ that is not essentially monotone.
Let $k=k_1+k_2$
and let $R$ be the Cartesian product of~$R_1$ and~$R_2$
as in the proof of Lemma~\ref{lem:numphard}.
Since $R$ is not affine, Theorem~\ref{thm:trichotomy}
(due to Dyer, Goldberg and Jerrum),
together with $\BIS \APred \SAT$, which follows from the fact that every problem in $\numP$
is AP-reducible to $\SAT$ \cite[Section 3]{DGGJ},
shows $\BIS \APred \numCSP{\{R\}}$.

The polynomial-time Turing reduction from $\numCSP{\{R\}}$ to  $\numMaxCSP{\{R\}}$
given in the proof of Lemma~\ref{lem:numphard} is actually an AP-reduction since
the satisfying assignments of~$J$ are in one-to-one correspondence with locally maximal satisfying assignments of~$J'$. So we have established $\BIS \APred \numMaxCSP{\{R\}}$.

Finally, $\numMaxCSP{\{R\}} \APred \numMaxCSP{\Gamma}$
since every instance of $\numMaxCSP{\{R\}}$ can be written as an instance of
$\numMaxCSP{\Gamma}$.
\end{proof}

{\renewcommand{\thetheorem}{\ref{lem:sathard}}
\begin{lemma} \statelemsathard{}\end{lemma}}
\begin{proof}Suppose that $R_1$ is an arity $k_1$ relation in $\Gamma$ that is not affine,
$R_2$ is an arity-$k_2$ relation in  $\Gamma$ that is not essentially monotone,
and $R_3$ is an arity $k_3$ relation in $\Gamma$ that is not in $\IMtwo$.
Let $k=k_1+k_2+k_3$
and let $R$ be the Cartesian product of~$R_1$ and~$R_2$ and~$R_3$.
Since $R_1$ is not affine and $R_3$ is not in $\IMtwo$,
$R$ is not affine and is not in $\IMtwo$. 
Thus, Theorem~\ref{thm:trichotomy}
 shows $\SAT \APred \numCSP{\{R\}}$.
Since $R_2$ is not essentially monotone, $R$ is not essentially monotone.
Thus, following the proof of Lemmas~\ref{lem:numphard} and~\ref{lem:bishard} we can use Lemma~\ref{lem:gadget} 
to obtain an AP-reduction from
$\numCSP{\{R\}}$ to $\numMaxCSP{\{R\}}$
and we can write every instance of $\numMaxCSP{\{R\}}$  as an instance of
$\numMaxCSP{\Gamma}$ to obtain 
an AP-reduction from $\numMaxCSP{\{R\}}$ to
$\numMaxCSP{\Gamma}$. 
\end{proof}

\bibliographystyle{plain}
\bibliography{\jobname}
 
\end{document}